 \documentclass[11pt]{article}
\normalsize
\usepackage{amsfonts}
\usepackage{amsmath}
\usepackage{amssymb}
\usepackage{latexsym}
\newtheorem{theorem}{Theorem}

\newtheorem{lemma}[theorem]{Lemma}

\newtheorem{Question}[theorem]{Question}

\newenvironment{proof}{\noindent  Proof:\ }{\hspace*{\fill} $\Box $\\}

\newcommand{\F}{\mathbb{F}}
\newcommand{\N}{\mathbb{N}}
\newcommand{\Z}{\mathbb{Z}}

\newcommand{\PSL}{\mbox{\rm PSL}}

\newcommand{\Aut}{\mbox{\rm Aut}}
\newcommand{\M}{\mbox{\rm M}}
\DeclareMathOperator{\rank}{rank}
\newcommand{\wt}{\mbox{\rm wt}}

\newcommand{\gen}{\mbox{\rm gen}}
\newcommand{\supp}{\mbox{\rm supp}}

\title{\bf On extremal  self-dual  codes of length 120}
\author{Javier de la Cruz\thanks{Javier de la Cruz is with the
  Universidad del Norte, Departamento de Matem{\'a}ticas,
 Km 5 V{\'i}a Puerto Colombia, Barranquilla, Colombia (e-mail:
jdelacruz@uninorte.edu.co).}
 }

\date{}

\begin{document}
\maketitle

\begin{abstract}
We prove that the only primes which may divide the order of the automorphism group
of a putative binary self-dual doubly-even $[120, 60, 24]$ code  are $2$, $3$, $5$, $7$, $19$, $23$ and $29$.
Furthermore we prove that automorphisms of prime order $p \geq 5$ have a unique cycle structure.
\end{abstract}


\section{Introduction}

Throughout the paper all codes are assumed to be binary and linear, if not explicitly stated  otherwise. Let
$C=C^\perp$ be a self-dual code of length $n$ and minimum distance
$d$. By results of Mallows-Sloane \cite{MS} and Rains \cite{Rains},
we have
\begin{equation} \label{cota-dual}  d \leq \left\{ \begin{array}{rl} 4 \lfloor \frac{n}{24} \rfloor + 4 & \mbox{if} \ n \not\equiv 22 \bmod \, 24 \\
       4 \lfloor \frac{n}{24} \rfloor + 6 & \mbox{if} \ n \equiv 22 \bmod \, 24,
   \end{array} \right.
\end{equation}
and $C$ is called extremal if equality holds. The length $n$ of an extremal self-dual doubly-even code $C$ is
bounded by $ n \leq 3928$, due to a result of Zhang \cite{Zhang}. Furthermore, if in addition $n=24m$,
then $C$ is always doubly-even, as shown by  Rains \cite{Rains}.

Already 1973
Sloane posed  the question whether  extremal
self-dual  codes of length $72$ exist  \cite{pregunta}. This is the first unsolved case
if $24 \mid n$.
Such codes are of particular interest since the supports of codewords of a given non-trivial
weight form a $5$-design according to the Assmus-Mattson Theorem  \cite{Assmus-articulo}.
Unfortunately, we know only two codes, the extended Golay code of length $24$ and the extended quadratic residue code of length $48$.
In order to find  codes of larger length
 non-trivial automorphisms may be helpful.
The following table shows what we know about the automorphism groups so far.  \\[2ex]
$$\begin{tabular}{c|c|c|l|l|c}
 parameters & codes &  $G$ & (possible) primes  & reference \\
   &  &  &  dividing $|G|$ &  \\

 \hline & & & &\\

[24,12,8] & ext. Golay & $\M_{24}$ & 2,3,5,7,11,{23} &  \cite{MacWilliams-Sloane77}   \\[1ex]

[48,24,12] & ext. QR & $\PSL(2,47)$ & 2,3,23,{47} &  \cite{Huff}, \cite{LPS}\\[1ex]

 [72,36,16] &   ? & $|G| \leq 24$ & 2,3,5 &  \cite{O'BW}, \cite{Nebe-Feulner7}  \\[1ex]

 [96,48,20] &   ? &  $|G| \leq $ ?    &  2,3,5 & \cite{Doncheva}, \cite{Disertacion}  \\[1ex]

\end{tabular}$$

\bigskip

Looking at the table one is naturally attempted to ask.

\begin{Question} {\rm Suppose that a self-dual $[120,60,24]$ code $C$ exist.
Are $2,3$ and $5$ the only primes which may divide the order of the
automorphism group of $C$? }
\end{Question}

If this turns out be true we have more evidence that for large $m$ the automorphism group
of an extremal self-dual code of length $n=24m$  may contain
only automorphisms of very small prime orders.
As a consequence such a code has almost no symmetries, i.e., it is  more or less a pure combinatorial
object and therefore probably hard to find if it exists.

In this paper we investigate
primes $p$ which may occur in the order of the automorphism group
of an extremal self-dual code of length $120$.
In Section III we prove that the only primes which may divide the order of the automorphism group
of  a putative binary self-dual doubly even  $[120, 60, 24]$  code $C$ are $2$, $3$, $5$, $7$, $19$, $23$ and $29$. Moreover, we exclude some cycle types of automorphisms of order 3, 5 and
7, which in particular shows that automorphisms of prime order $p \geq 5$ have a unique cycle structure.
For involutions the possible cycle types are known by \cite{Stefka}.
Thus, as the main result we obtain

\begin{theorem} \label{main} Let $C$ be an extremal self-dual code of length $120$.
\begin{itemize}
\item[\rm a)] The only primes with may divide the order of the automorphism group
of $C$ are $2,3,5,7, 19, 23$ and $29$.
\item[\rm b)] If $\sigma$ is an automorphism of $C$ of prime order $p$ then $p=2,3,5,7,19,23$ or $29$ and its cycle structure is given by
\begin{center} {
\begin{tabular}{c|c|c}
p & \mbox{\rm number of } & \mbox{ \rm number of } \\
 & \mbox{\rm  $p$-cycles} & \mbox{ \rm fixed points} \\
\hline
$2$ & $48,60$ & $24,0$ \\
$3$ & $ 32,34,36,38,40$ & $24,18,12,6,0 $\\
$5$ & $24$ & $0$ \\
$7$ & $17$ & $1$ \\
$19$ & $6$ & $6$ \\
$23$ & $5$ & $5$ \\
$29$ & $4$ & $4$  \\
\end{tabular} }
\end{center}
\end{itemize}
\end{theorem}

\mbox{} \\

In a forthcoming paper we will prove that automorphisms of order $3$ act fixed point freely.
Thus apart from (possibly) involutions all elements in $\Aut(C)$ of prime order have a unique cycle structure. \\[-2ex]

Theorem \ref{main} is part of my PhD thesis \cite{Disertacion}.

\section{Preliminaries}
Let $C$ be a binary code with an automorphism $\sigma$ of odd prime
order $p$. If $\sigma$ has $c$ cycles of length $p$ and $f$ fixed
points, we say that $\sigma$ is of type $p$-$(c;f)$.
Without loss of generality we may assume that
$$\sigma=(1,2, \ldots, p)(p+1,p+2, \ldots , 2p) \ldots
((c-1)p+1,(c-1)p+2, \ldots ,cp).$$
Let $\Omega_{1},\Omega_{2},\ldots,\Omega_{c}$ be the cycle sets and let
$\Omega_{c+1},\Omega_{c+2},\ldots,\Omega_{c+f}$ be the fixed points
of $\sigma$.
We put $F_{\sigma}(C)=\{v\in C\mid v\sigma=v\}$ and
$ E_{\sigma}(C)=\{v\in C\mid \wt(v|_{\Omega_{i}})\equiv 0\,\,
\mbox{mod}\,\,2, i=1,\ldots,c+f\},$ where $v|_{\Omega_{i}}$ is the
restriction of $v$ on $\Omega_{i}$. With this notation we have

\begin{lemma} {\rm (\cite{Huff}, Lemma 2)} \label{l1} \quad  $C=F_{\sigma}(C)\oplus
E_{\sigma}(C)$.
\end{lemma}

There is an obvious relation between the weight distribution of $C$ and its subcode $F_\sigma(C)$, namely

\begin{lemma} \label{relation} If $A_i$ and $A'_i$ denotes the number of codewords of weight $i$ in $C$ resp. $F_\sigma(C)$ then
$A_i \equiv A'_i \bmod p$.
\end{lemma}
\begin{proof} If $c \in C$ and $c \not\in F_\sigma(C)$ then the size of the orbit of $c$ under $\sigma$ is divisible by $p$.
\end{proof}

Clearly, a generator matrix of  $C$ can be
written in the form
$$
{\rm gen}(C)=   \left(%
\begin{array}{cc}
  X & Y \\
  Z & 0 \\
\end{array}%
\right)
\begin{array}{c}
  \} \ \text{{\rm gen}}(F_\sigma(C))\\
  \} \ \text{{\rm gen}}(E_\sigma(C)), \\
\end{array}
$$
where the first part of the matrix correspond to all coordinates which are moved by $\sigma$ and the second to the $f$ fixed points. \\

If $\pi:F_{\sigma}(C)\rightarrow \mathbb{F}_{2}^{c+f}$ denotes the map
defined by $\pi(v|_{\Omega_{i}})=v_{j}$ for some $j\in \Omega_{i}$ and
$i=1,2,\ldots,c+f$, then
 $\pi(F_{\sigma}(C))$ is a binary $[c+f,\frac{c+f}{2}]$ self-dual code (\cite{Huff}, Lemma 1).\\

Note that every binary vector of length $p$ can be
identified with a unique polynomial in the factor algebra $\mathbb{F}_{2}[x]/(x^{p}-1)$ by
$(v_{0},v_{1},\ldots,v_{p-1}) \mapsto
v_{0}+v_{1}x+\ldots+v_{p-1}x^{p-1}\in \F_2[x].$
Furthermore, recall that  the vector space of
even-weight polynomials in $\mathbb{F}_{2}[x]/(x^{p}-1)$, which we denote by $P$,
 is a binary cyclic code of length $p$ generated by $x-1$.
Let $E_{\sigma}(C)^{*}$ be the subcode of $E_{\sigma}(C)$ where the last
$f$ coordinates have been deleted. For $v\in E_{\sigma}(C)^{*}$ we may
consider the $p$-cycle
$$v\mid\Omega_{i}=(v_{0},v_{1},\ldots,v_{p-1}) \qquad (i=1, \ldots, c)   $$
as the polynomial
$$\varphi(v\mid\Omega_{i})(x)=v_{0}+v_{1}x+\ldots+v_{p-1}x^{p-1}$$
in $P$. In this way we obtain a map
$\varphi:E_{\sigma}(C)^{*}\rightarrow P^{c}$. Clearly,
$\varphi(E_{\sigma}(C)^{*})$ is a submodule of the $P$-module
$P^{c}$. If  the multiplicative order of 2 modulo $p$, usually denoted by $s(p)$, is $p-1$, then
the check polynomial $1+x+x^{2}+\ldots+x^{p-1}$ of  $P$ is
irreducible over $\mathbb{F}_{2}$. Hence $P$ is an extension field of $\F_2$ with identity $e(x)=x + \ldots +x^{p-1}$ and $\varphi(E_{\sigma}(C)^{*})$ is
a code over the field $P$.

\begin{lemma} {\rm (\cite{Yorgov}, Theorem 3)} \label{principal2} Assume that $s(p)=p-1$.
Then a binary code $C$ with an automorphism $\sigma$ of odd prime order $p$ is
self-dual if and only if the following two conditions hold.
\begin{enumerate}
    \item[\rm a)] $\pi(F_{\sigma}(C))$ is a binary self-dual  code
    of length $c+f$.
    \item[\rm b)] $\varphi(E_{\sigma}(C)^{*})$ is a self-dual code of
    length $c$ over the field $P$ under the inner product $u \cdot v=\sum_{i=1}^{c}
    u_{i}v_{i}^q$ for $q={2^{\frac{p-1}{2}}}$.
\end{enumerate}
\end{lemma}

\begin{lemma} {\rm (\cite{Yo56}, Theorem 3)} \label{operaciones}
Let $C$  be  a binary self-dual code  and
let  $\sigma$ be an
automorphism of  $C$ of odd prime order $p$.
Then any  Code,  which can be obtained from $C$ by
\begin{itemize}
  \item[\rm (i)] a substitution $ x \rightarrow x^t$ in $\varphi(E_{\sigma}(C)^{*})$, where $t$ is an integer
            with $1 \leq t \leq p-1$ or
  \item[\rm (ii)] a multiplication of the $j$th coordinate of
    $\varphi(E_{\sigma}(C)^{*})$ by $x^{t_j}$, where $t_{j}$ is an
    integer with $0\leq t_{j}\leq p-1$ and  $j=1,2,\ldots,c$,
  \end{itemize}
is equivalent to $C$.
\end{lemma}

Lemma \ref{principal2} and Lemma \ref{operaciones} are crucial to exclude the prime $59$  in the order of the automorphism group
of an extremal self-dual code of length $120$. Most of the other primes
can be  excluded by the following two results.

\begin{lemma} {\rm (\cite{BMW}, Theorem 7)} \label{BMW1} If $C$ is a binary extremal self-dual code of length $24m+2r$ where $ 0 \leq r \leq 11, \, m \geq 2$,
and $\sigma$ is an automorphism of $C$ of type $p$-$(c;f)$ for some prime $p \geq 5$ then $c \geq f$.
\end{lemma}

\begin{lemma} {\rm (\cite{Yo56}, Theorem 4)} \label{cases}
Let $C$ be a binary self-dual $[n,k,d]$ code and let
$\sigma\in\Aut(C)$ be of type $p$-$(c;f)$ where $p$ is an odd
prime. If $g(s)=\sum_{i=0}^{s-1}\lceil\frac{d}{2^i}\rceil$ then
\begin{itemize}
\item[\rm a)] $pc\geq g(\frac{p-1}{2}c)$ \ \mbox{and}
\item[\rm b)] $f\geq g(\frac{f-c}{2})$ for $f>c$.
\end{itemize}
\end{lemma}

\begin{lemma}  {\rm (\cite{BMW}, Lemma 4)}\label{BMW2}  Let $C$ be a binary self-dual code of length $n$
and let $\sigma$ be an automorphism of $C$ of type $p$-$(c;f)$ where $p$ is an odd prime.
If $s(p)$ is even, then $c$ is even. 
\end{lemma}

Let $C_{\pi_1}$ be the subcode of $\pi(F_\sigma(C))$ which consists of
all codewords which have support in the first $c$ coordinates and let $C_{\pi_2}$ be the subcode of $\pi(F_\sigma(C))$ of all codewords which have support in the last $f$ coordinates. Then $\pi(F_\sigma(C))$ has a generator matrix of the form
\begin{equation}\label{form-matrix}{\rm{gen}}(\pi(F_{\sigma}(C)))=\left(%
\begin{array}{cc}
  A & O \\
  O & B \\
  D & E \\
\end{array}%
\right),
\end{equation}
where $ (A \,O)$ is a generator matrix of $C_{\pi_1}$ and $(O\, B)$ is a
generator matrix
 of $C_{\pi_2}$, $O$ being the appropriate size zero matrix.
With this notation we have

\begin{lemma} {\rm (\cite{pless-libro}, Theorem 9.4.1)} \label{pless} If $k_{1}=\dim\, C_{\pi_1}$ and $k_{2}=\dim\, C_{\pi_2}$, then the following holds true.
\begin{enumerate}
\item[\rm a)] {\rm(}Balance Principle{\rm)} $k_{1}-\frac{c}{2}=k_{2}-\frac{f}{2}$.
\item[\rm b)] $\rank D =\rank E =\frac{c+f}{2}-k_{1}-k_{2}.$
\item[\rm c)] Let $\mathcal{A}$ be the code of length $c$ generate by $A$, $\mathcal{A}_{D}$ the code of length $c$ generated by
 the rows of $A$ and $D$, $\mathcal{B}$ the code of length $f$ generated by $B$, and $\mathcal{B}_{E}$ the code of length
  $f$ generated by the rows of $B$ and $E$. Then $\mathcal{A}^\perp=\mathcal{A}_{D}$ and $\mathcal{B}^\perp=\mathcal{B}_{E}$.
\end{enumerate}
\end{lemma}

\begin{lemma} \label{matrix-I} Let $C$ be a binary self-dual code with minimum distance $d$ and let $\sigma \in \Aut(C)$ be of type $p$-$(c;f)$ with $c=f<d$. Then $\pi(F_{\sigma}(C))$ has a generator matrix of the form $(I\mid E')$ where $I$ is the identity matrix of size $c$.
\end{lemma}
\begin{proof} We may write ${\rm{gen}}(\pi(F_{\sigma}(C)))$ as in (\ref{form-matrix}).
The condition $f<d$ implies that $B=0$. Since $c=f$, by the Balance Principle we see that $A=0$. Thus $D$ is regular and
$$ D^{-1}{\rm{gen}}(\pi(F_{\sigma}(C)))= (I \mid E')$$
is a generator matrix of $\pi(F_{\sigma}(C))$.
\end{proof}

\begin{lemma} \label{BMW3}  If $p$ is an odd prime and  $C$ is a binary self-dual code with minimum distance $d$,
then $\Aut(C)$ does not contain an automorphism  $\sigma$ of type
$p$-$(c;f)$  with $c=f$ and $p+c < d$.
\end{lemma}
\begin{proof}  The condition $p+c< d$ implies $c=f< d$ and by Lemma \ref{matrix-I}, we obtain a
generator matrix of the form
$$ {\rm gen}(\pi(F_\sigma(C))) = (I \mid E').$$ Let $v$ be any row
vector of $(I \mid E')$. Then $\pi^{-1}(v) = c \in C$ has weight
$$  \wt(c) \leq p + f = p + c < d,$$
a contradiction.
\end{proof}

\section{Excluding primes in the automorphism group of  extremal self-dual codes
of length 120}
In this section  $C$ always denotes a binary self-dual code with parameters $[120,60,24]$. The weight enumerator of the code $C$ is determined in \cite{MS} as
\begin{equation} \label{enumerator} W_{C}(y)=1+39703755y^{24}+6101289120y^{28}+475644139425y^{32}+\ldots\end{equation}  Suppose that there is a  $\sigma \in \Aut(C)$
 of prime order $p \geq 3$. According to Lemma  \ref{BMW1},  \ref{cases} and  \ref{BMW2},  the only possibilities for the type of $\sigma$ are
the following.

\begin{equation} \label{tipos}
\begin{tabular}{c|c|c}
p & c & f \\
\hline
3 & 30, 32, 34, 36,  & 30, 24, 18, 12, \\
&  38, 40 &   6, 0 \\
5 & 20,22,24 & 20,10,0 \\
7 & 15,16,17 & 15,8,1 \\
11 & 10 & 10 \\
13 & 9 & 3 \\
17 & 7 & 1 \\
19 & 6 & 6 \\
23 & 5 & 5 \\
29 & 4 & 4  \\
59 & 2 & 2
\end{tabular}
\end{equation}
\bigskip

\begin{lemma}\label{primo11} The primes $p=11, 13$ and $p=17$ do not divide $|\Aut(C)|$.
\end{lemma}

\begin{proof}
If $p=11$, then $c=f=10$ and $p+c = 11 +10 =21 < d = 24$. Now we apply Lemma \ref{BMW3} to exclude $p=11$.

For $p=13$ we have  $s(p)=12$. Thus, by Lemma \ref{BMW2}, $c$ must be even, which contradicts $c=9$. Hence $p=13$
is not possible.

Note that $s(17)$ is even. Applying again Lemma \ref{BMW2} we obtain $c$
 even, which contradicts $c=1$. Thus $p=17$ does not occur either.
\end{proof}

\begin{lemma} \label{para59} The prime $59$ does not divide $|\Aut(C)|$.
\end{lemma}

\begin{proof} Suppose that $\sigma \in \Aut(C)$ is of order $59$. Thus $\sigma$ is of type $59$-$(2;2)$.
 We determine all possible generator matrices for
$C$ with respect to the decomposition given in Lemma \ref{l1} and check by computer that the minimum distance  is smaller than $24$ in each
case.\\[-2ex]

\noindent
Step 1: Construction of all possible $ \gen(C)$.\\
By Lemma \ref{principal2}, the code $ \pi(F_\sigma(C))$ is  self-dual and has parameters $[4,2,2]$. Thus
$$
{\rm gen}(\pi (F_{\sigma}(C)))=\left(%
\begin{array}{cc|cc}
  1 & 0 & 1 & 0 \\
  0 & 1 & 0 & 1 \\
\end{array}%
\right).
$$
Consequently,
$$
{\rm gen}(F_{\sigma}(C))=\left(%
\begin{array}{cc|cc}
  \textbf{1} &\textbf{0}& 1 & 0 \\
  \textbf{0} &\textbf{1}& 0 & 1 \\
\end{array}%
\right)
$$
where \textbf{1} is the all-one vector and  \textbf{0} the zero-vector of length $59$.

Next we determine $\gen(E_\sigma(C))$. Note that $s(59)=58$. Thus,
by Lemma \ref{principal2}, the vector space
 $\varphi(E_\sigma(C)^*)$ is a self-dual $[2,1]$ code over the field $P =\F_{2^{58}}$ under the inner product $u\cdot v=u_{1}v_{1}^{2 ^{29}}+u_{2}v_{2}^{2 ^{29}}$. W.l.o.g. it is generated  by some
vector
$(e(x),b(x)) \in P^2 $
where $e(x)$ denotes the identity in $P$ and
$$   e(x) + b(x)^{2^{29} +1} =0.$$
Let $\alpha(x)$ be a generator of the multiplicative group of $P$.
Writing $b(x)= \alpha(x)^r$ with $0 \leq r \leq 2^{58}-2 $ we obtain
$$ \alpha(x)^{r(2^{29} +1)} = e(x).$$
Thus $r=(2^{29}-1)k$ for some $k \in \mathbb{N}_0$ and therefore
$$ b(x) = \alpha(x)^r = (\alpha(x)^{2^{29}-1})^k = \delta(x)^k $$
where $ \delta(x) = \alpha(x)^{2^{29}-1}$ and $ 0 \leq k \leq 2^{29}$.
It follows, by Lemma \ref{l1}, that $C$ has a generator matrix of the form
$$ {\rm gen}(C)=\left(%
\begin{array}{cc|cc}
 \textbf{1} & \textbf{0} & 1 & 0 \\
 \textbf{0} &  \textbf{1}& 0 & 1 \\
 \hline
  [e(x)] & [\delta(x)^{k}] & 0 & 0 \\
\end{array}%
\right) \;\;\;(*)
$$
where $[e(x)]$ and $[\delta(x)^{k}]$ are circulant $58\times
59$-matrices and $0 \leq k \leq 2^{29}$.
We would like to mention here that some of the $2^{29}+1$ generator matrices may
define equivalent codes. \\[-2ex]

\noindent
Step 2: Reduction of the number of  generator matrices $\gen(C)$ in step 1.\\
Observe that $\langle \delta(x)\rangle$ is a subgroup of $P\backslash \{ 0\}$ of order $2^{29}+1 = 3 \cdot 59 \cdot 3033169$
which contains the subgroup $ \langle xe(x) \rangle $ of order $59$.
Since the factor group $\langle \delta(x)\rangle/\langle xe(x) \rangle$ is generated by $\langle x e(x)\rangle\delta(x)$
we obtain
$$ \langle \delta(x) \rangle = \cup_{i=1}^{3 \cdot 3033169}\  \langle xe(x) \rangle \delta(x)^i. $$If we multiply $\delta(x)^k$ in step 1 by $x^t$ for some $ 0 \leq t \leq 58$, the corresponding generator matrix
defines an equivalent code by Lemma \ref{operaciones}, part (ii).
Thus, in $(*)$, we only have to consider the polynomials
$$ \delta(x)^k \;\; \mbox{for} \;\;   k=1, \ldots, 9099507.$$

Next we apply the substitution $ x \rightarrow x^2$ in $\varphi(E_\sigma(C)^*)$ which also
leads to an equivalent code by Lemma \ref{operaciones}, part (i).
Clearly, this substitution applied to the generator $(e(x), \delta(x)^k)$ yields
$$  (e(x^2), \delta(x^2)^k) = (e(x),\delta(x)^{2k}).$$
Now we divide $\Z_{9099507}$ into a disjoint union of orbits
$$  {\rm orb}(i) = \{2^n i \bmod 9099507 \mid i \in \N_0 \}.$$
Observe that for all $ j \in {\rm orb}(i)$ the corresponding codes $C$ are equivalent.
With Magma one easily checks that there are exactly $156889$ orbits.
This shows that we only have to consider generator matrices
$$ {\rm gen}(C)=\left(%
\begin{array}{cc|cc}
 \textbf{1} & \textbf{0} & 1 & 0 \\
 \textbf{0} &  \textbf{1}& 0 & 1 \\
 \hline
  [e(x)] & [\alpha(x)^{t (2^{29}-1)}] & 0 & 0 \\
\end{array}%
\right)
$$
where $t$ runs  through a set of representatives of the $156889$ orbits.
In each case we find with Magma a codeword of minimum distance smaller $24$
which completes the proof.
\end{proof}

So far we have proved part a) of the Theorem.

\section{The cycle strucures}

In this section we prove part b).

\begin{lemma} \label{simetrico} Let $C$ be a self-dual $[120, 60, 24]$ code. Then $C$ has no automorphism of type $p$-$(c;c)$ for $p=3,5$
and $7$.
\end{lemma}
\begin{proof}
According to the list in (\ref{tipos}) we have to show that $C$ does not have an automorphism of type $3$-$(30;30)$, $5$-$(20;20)$ or $7$-$(15;15)$. \\[-2ex]

\noindent
Claim 1: $C$ has no automorphism of type $3$-$(30;30)$.\\
Let $\sigma \in \Aut(C)$ of type $3$-$(30;30)$. Then  $\pi(F_\sigma(C))$
is a self-dual $[60, 30, d_{\pi}]$ code. According to
(\ref{cota-dual}) of the introduction we have $d_{\pi}\leq 12$. Now we take a generator matrix for $\pi(F_\sigma(C))$ in the form of
(\ref{form-matrix}). Since $c=f$, we get
$k_{1}=k_{2}$,  by the Balance Principle (see Lemma \ref{pless}). Note that $\mathcal{B}$ is a doubly even  $[30,
k_{2}, d']$ code with $d'=24$ or $d'=28$.
If $k_2 \geq 2$, then obviously $\pi(F_\sigma(C))$, and therefore $C$ contains a codeword
of weight less or equal $12$, a contradiction. Thus $k_1 = k_2 \leq 1$.

If $k_{2}=0$, then ${\rm{gen}}(\pi(F_{\sigma}(C)))=(I_{30}\,|\,E)$.
 Let $(e_i\,|\,v_i)$ denote the $i$-th row of $(I_{30}\,|\,E)$. Since
$\wt(\pi^{-1}(e_i|v_i))=3+\wt(v_i)\geq 24$, we get
$\wt(v_i)=21, 25$ or $29.$ If $\wt(v_i)=29$ and $\wt(v_j)=29$,
then
 $$ S_{v_i,v_j} = |\, \supp(v_i) \cap \supp(v_j)\,| \, \geq \, 28$$
and therefore $\wt(\pi^{-1}(e_i +e_j|v_i +v_j))=6+\wt(v_i +v_j)\leq 8$,
a contradiction.
In all other cases we get similarly a contradiction unless
 $\wt(v_i)=21$ for all $i=1, \ldots, 30$.

If $x=(e_i|v_i)$ and $y=(e_j|v_j)$, then $S_{x,y}=S_{v_{i},v_{j}} \geq 12$.
In case  $S_{v_i,v_j} > 12$ we obtain  $\wt(\pi^{-1}(x+y)) \leq 6 + 17  = 23$, a contradiction.
 Consequently $S_{v_{i},v_{j}}=12$ for all $i\neq j \in \{1, \ldots,
30 \}$. Hence two vectors $v_{i},v_{j}$ do not have a coordinate
simultaneously zero. This implies that the dimension of ${\rm{gen}}(\pi(F_{\sigma}(C)))$
is at most $3$, a contradiction.

If $k_{2}=1$, then $\pi(F_\sigma(C))$ has a generator
matrix of the form
$$\left(%
\begin{array}{cc}
  a & 0\ldots 0 \\
  0\dots0 & b \\
  D & E \\
\end{array}%
\right),
$$
where $\wt(b) =24$ or $28$. Since $C$ is doubly even, $\wt(a) \in \{8, 12, 16, 20, 24, 28\}$.
Suppose that $\wt(a)=28$. Then $\wt(\pi^{-1}(a|b)) \geq 108$ which implies that $(a|b)$ is the all-one vector,
 a contradiction to $\wt(a)=28$. Thus $\wt(a) \leq 24$.
Therefore $a$ contains in at least 6 positions $0$. Consequently,
there are at least 6 vectors of the form $z_{i}=(0,0,\ldots, 1,\ldots, 0,0) \in \F_2^{30}$, which are orthogonal to $a$.
By Lemma \ref{pless} c), we obtain
 $z_{i}\in \mathcal{A}^\perp=\mathcal{A}_{D}$.
The contradiction now follows as in case $k_2=0$.\\[-2ex]

\noindent
Claim 2: $C$ has no automorphism of type $5$-$(20;20)$.\\
Note that $p=5\equiv 1 \bmod 4$. Thus, by (\cite{Huff}, Lemma 1), the space $\pi(F_{\sigma}(C))$ is a doubly even self-dual $[40, 20, d_{\pi}]$ code.
Furthermore $c=f=20<d$. According to Lemma \ref{matrix-I} we can take a generator matrix of $\pi(F_{\sigma}(C))$ of the form
${\rm gen}(\pi(F_{\sigma}(C)))=(I \mid E')$.
 If $x=(1,0,0, \ldots, 0 \mid e_{1})$ and $y=(0,1,0, \ldots, 0 \mid e_{2})$ denote the first respectively the second row of $(I \mid E')$,
 then $$\wt(\pi^{-1}(x))=\wt(\pi^{-1}((1,0,0, \ldots, 0 \mid e_{1}))=5+\wt(e_{1})\geq24.$$
Therefore $19\leq \wt(e_{1})\leq 20$. Since $C$ and $\pi(F_{\sigma}(C))$ are doubly even we obtain
$\wt(e_{1})=19$.
Similarly $\wt(e_2)=19$.
This implies that $\wt(e_{1}+e_{2})\leq 2$. Hence $$\wt(\pi^{-1}(x+y))=\wt(\pi^{-1}(1,1,0, \ldots, 0 \mid e_{1}+e_{2}))=2\cdot 5+\wt(e_{1}+e_{2})\leq12,$$ a contradiction.

\noindent
Claim 3: $C$ has no automorphism of type $7$-$(15;15)$.\\
Since $c=f=15$ and $p+c = 7 +15 =22 < d = 24$ we may apply Lemma \ref{BMW3} to see that there is no automorphism
of type  $7$-$(15;15)$.
\end{proof}

\begin{lemma} \label{5-cycles} A self-dual $[120,60,24]$ code does not have an automorphism of type $5$-$(22;10)$.
\end{lemma}
\begin{proof} Suppose that $\sigma \in \Aut(C)$ is of type $5$-$(22;10)$. Then $\pi(F_\sigma(C))$ is a
self-dual $[32,16,d_\pi]$ code. Furthermore, $\pi(F_\sigma(C))$ is doubly even, by (\cite{Huff}, Lemma 1), since
$ p \equiv 1 \bmod 4$. According to (\ref{cota-dual}), we have $d_\pi \leq 8$. If we write
 $d_\pi = x +y $ where
$x$ is the number of $1$s in the first $c=22$ coordinates of a minimal weight codeword and $y$ is the number of $1$s in the last $f=10$ coordinates,
 then $x+y \leq 8$ and $5x + y \geq 24$.
 This forces $ x \geq 4$ and $d_\pi=8$. Hence $\pi(F_\sigma(C))$ is an extremal
self-dual doubly even code of length $32$. By (\cite{RS}, p. 262) there are (up to isometry) exactly five such codes,
denoted by $C81$ (extended quadratic residue code), $C82$ (Reed-Muller code), $C83$, $C84$ and $C85$.
To see that no one of these codes can occur as $\pi(F_\sigma(C))$ we proceed as follows.

Let $C_0$ denote one of these code. We do not know which coordinates belong to the fixed points of
$\sigma$. We know only the number, namely $10$. Therefore we choose all possible $10$-subsets of
$1, \ldots, 32$ and take them as the coordinates of fixed points.
In each case we construct $\pi^{-1}(C_0)$ and compute the minimum distance with MAGMA.
In turns out that all distances are strictly less than $24$. Thus none of the five
extremal doubly even codes of length $32$ can occur as $\pi(F_\sigma(C))$, a contradiction.
\end{proof}


\begin{lemma} $C$ has no automorphism of type $7$-$(16;8)$.
\end{lemma}
\begin{proof}
Let $\sigma$ be an automorphism of type $7$-$(16;8)$. Then  $\pi(F_\sigma(C))$
is a self-dual $[24, 12, d_{\pi}]$ code. According to
(\ref{cota-dual}) we have $d_{\pi}\leq 8$. If $d_{\pi}=x+y$ where
$x$ is again the number of $1$s in the left $16$ coordinates  and $y$ is the
number of $1$s in the right $8$  coordinates of a codeword of minimal weight, then $x+y\leq 8$ and
$7x+y\geq 24$. Therefore $x\geq 3$ and $d_{\pi}=4, 6 $ or $ 8$. In total there are $30$
 self-dual $[24, 12, d_{\pi}]$ codes (see \cite{PlessSloane}, \cite{Clasification}),
 one with $d_{\pi} = 8$, one with $d_{\pi} = 6$ and $28$ with $d_{\pi} = 4$.\\
If $d_{\pi}=8$ then $\pi(F_\sigma(C))$ is the Golay code. The
weight enumerator of the Golay code is $1+759y^8+2576y^{12}+759y^{16}+y^{24}.$
We know that a vector of $F_\sigma(C)$ of weight $28$ can be
formed only by vectors of $\pi(F_\sigma(C))$ of weight $4$ and
$10$ since  $28=4 \cdot 7+0$ and $28=3 \cdot 7+7$. Therefore, $F_\sigma(C)$ has no codewords of
weight $28$. But this contradicts the fact that the number $A_{28}$ (see (\ref{enumerator})) of codewords of $C$ of weight $28$  satisfies $A_{28}=6101289120\equiv 3\,\,
\mbox{mod}\,\,7$, by Lemma \ref{relation}.\\
If $d_{\pi}=6$ then $\pi(F_\sigma(C))$ is the code $Z_{24}$ (see \cite{Clasification}, TABLE E).
 In this case we take all possibilities for the $8$ fixed points and construct $\pi^{-1}(Z_{24})$. In all situations we find with MAGMA a vector of weight  less than $24$ or not divisible by $4$.\\
Thus we are left with the case $d_\pi=4$.
Now observe the following fact.
If a vector of $\pi(F_\sigma(C))$ has weight $4$, then all non-zero coordinates correspond to cycles,
since $C$ has minimum distance $24$.
 So, if $\pi(F_\sigma(C))$ has  components $d_{n}$ or $e_{n}$ (for notation see \cite{PlessSloane}), then the corresponding coordinates are cycles.
 With this observation we easily see that $\sigma$ has less than $8$ fixed points unless $\pi(F_\sigma(C)$ is of type
 $X_{24}$ or $Y_{24}$. The case $ \pi(F_\sigma(C))= X_{24}$ can not occur since it yields a vector of weight $30$ in $C$.
The final case $Y_{24}$ has been excluded with MAGMA similar to the case $Z_{24}$.
\end{proof}

\noindent
{\bf Acknowledgment} The author would like to thank Professor Willems for their contributions and valuable suggestions.

\end{document}